\documentclass[letterpaper,12pt]{article}

\usepackage{amsthm}
\usepackage{amsfonts}
\usepackage{amsmath}
\usepackage{graphicx}
\usepackage{epstopdf}
\usepackage{lineno}
\usepackage{setspace}
\usepackage{verbatim}
\usepackage{authblk}

\newtheorem{theorem}{Theorem}
\newtheorem*{theorem*}{Theorem}

\newtheorem{corollary}[theorem]{Corollary}

\newtheorem{lemma}[theorem]{Lemma}

\newtheorem{observation}{Observation}
\newtheorem{proposition}[theorem]{Proposition}

\newcommand{\gbar}{\overline{G}}

\title{Optimal antithickenings of claw-free trigraphs}
\author{Maria Chudnovsky\thanks{Supported by NSF grants DMS-1001091 and IIS-1117631.}}
\affil{Department of Industrial Engineering and Operations Research\\Columbia University, New York NY}
\author{Andrew D.~King\thanks{Corresponding author.  Email: andrew.d.king@gmail.com.  Supported by an EBCO/Ebbich Postdoctoral Scholarship and the NSERC Discovery Grants of Pavol Hell and Bojan Mohar.}}
\affil{Departments of Mathematics and Computing Science\\Simon Fraser University, Burnaby BC}

\newenvironment{enumerate*}{
\begin{enumerate}
  \setlength{\itemsep}{5pt}
  \setlength{\parskip}{0pt}
  \setlength{\parsep}{0pt}
}{\end{enumerate}}

\newenvironment{itemize*}{
\begin{itemize}
  \setlength{\itemsep}{5pt}
  \setlength{\parskip}{0pt}
  \setlength{\parsep}{0pt}
}{\end{itemize}}

\newenvironment{proof*}[1][\proofname]{%
\proof
}{\endproof}

\begin{document}

\maketitle

\begin{abstract}
Chudnovsky and Seymour's structure theorem for claw-free graphs has led to a multitude of recent results that exploit two structural operations:  {\em compositions of strips} and {\em thickenings}.
In this paper we consider the latter, proving that every claw-free graph has a unique optimal {\em antithickening}, where our definition of {\em optimal} is chosen carefully to respect the structural foundation of the graph. Furthermore, we give an algorithm to find the optimal antithickening in $O(m^2)$ time.  For the sake of both completeness and ease of proof, we prove stronger results in the more general setting of trigraphs.
\end{abstract}


\newpage

\section{Claw-free and quasi-line graphs and trigraphs}

The structural characterization of claw-free graphs \cite{clawfree5} and quasi-line graphs \cite{clawfree7} has led to a recent explosion of results of the same general type:  Take some statement $\mathcal A$ that we know to hold for line graphs and circular interval graphs.  Do we know that $\mathcal A$ holds for quasi-line graphs?  What about claw-free graphs?  And if $\mathcal A$ guarantees the existence of a certain combinatorial object (e.g.\ a $\lceil \frac 32\omega(G)\rceil$-colouring of a quasi-line graph $G$ \cite{chudnovskyo07}), how efficiently can we find such an object?

For quasi-line graphs, the structure theorem essentially tells us that all quasi-line graphs can be built in the following way:
\begin{enumerate*}
\item[i)] Take a circular interval graph or a {\em composition of linear interval strips}, which is constructed by ``replacing'' every edge of a multigraph with a linear interval graph (in the same way that we would replace every edge with a vertex when constructing a line graph).
\item[ii)] Take a matching such that removing a submatching does not change the fact that the graph is quasi-line, and expand the endpoints of each edge of the matching into a {\em homogeneous pair of cliques}.
\end{enumerate*}

For claw-free graphs the situation is more complicated, but at its heart very similar (at least when $\alpha\geq 4$).

The standard (and seemingly almost universal) proof method for quasi-line graphs can be roughly described as follows.  First, prove that a minimum counterexample contains no {\em nonlinear homogeneous pair of cliques}.  Second, note that since the statement holds for circular interval graphs, the structure theorem for quasi-line graphs \cite{clawfree7} tells us that a minimum counterexample must be a composition of linear interval strips.  Third, find a {\em linear interval strip decomposition} of a supposed minimum counterexample and, using line graphs as a base case, apply induction on the size of the graph.  Again, for claw-free graphs the approach is more complicated but fundamentally the same.

In the past few years, this approach has yielded results that include, among others,
\begin{enumerate*}
\item Any quasi-line graph with chromatic number at least $k$ contains $K_k$ as a minor (Hadwiger's conjecture holds for quasi-line graphs) \cite{chudnovskyo08}.
\item We can find a maximum-weight independent set in a quasi-line graph in $O(n^3)$ time \cite{faenzaos11}.
\item The facet-defining inequalities of the stable set polytope of a quasi-line graph fall into several well-understood categories \cite{cssurvey, eisenbrandosv05}.
\item The problem {\em Minimum Dominating Set} is fixed-parameter tractable for quasi-line graphs \cite{hermelinmlw11}.
\item Any quasi-line graph with chromatic number at least $k$ contains a vertex $v$ for which $\lceil \frac 12 (d(v)+2+\omega(G[N(v)])\rceil \geq k$ (the {\em local strengthening} of Reed's conjecture holds for quasi-line graphs) \cite{chudnovskykps11,kingthesis}.
\item Any quasi-line graph can be coloured in polynomial time using at most $\chi_f+3\sqrt{\chi_f}$ colours \cite{kingr11}.
\end{enumerate*}
Statements 2, and 4, as well as weakenings of the other statements, are known to hold for claw-free graphs \cite{chudnovskyo08, gallucciogv10, kingthesis, kingr11}.

Any algorithmic application of this approach requires the ability to do two things efficiently: find and reduce nonlinear homogeneous pairs of cliques, and find a linear interval strip decomposition of any composition of linear interval strips.  In this paper we concern ourselves with the first issue.  The fact that we can do this in polynomial time was first noted independently by King and Reed \cite{kingr08} and Oriolo, Pietropaoli, and Stauffer \cite{oriolops08}.  In this paper we improve upon these results by giving a faster algorithm, in the more general setting of trigraphs, and proving that there is a unique optimal reduction.  As it turns out, we need to consider trigraphs in order to sensibly define and prove the uniqueness of {\em optimal antithickenings}.  Faenza, Oriolo, and Snels \cite{faenzaos11b} independently proved a similar result in the setting of graphs using a multi-step reduction, without considering optimality or uniqueness.  One consequence of our one-step reduction is that the running time of the algorithmic results in Corollaries 21 and 22 of \cite{faenzaos11b} can easily be improved by a factor of $m$.

We now state the main result of the paper, while deferring formal definitions until the end of this section.  We also state a corollary that some readers may find more digestible and useful.  As usual, $m$ and $n$ denote the number of edges and vertices in a graph (or trigraph, in which case $m$ is the number of adjacent unordered vertex pairs) respectively.

\begin{theorem}\label{thm:anti}
Let $G$ be a quasi-line trigraph that is not cobipartite, or a claw-free trigraph with $\alpha(G)\geq 3$.  Then $G$ has a unique optimal antithickening, and we can find it in $O(m^2)$ time.
\end{theorem}

\begin{corollary}\label{cor:anti}
Let $G$ be a quasi-line graph that is not cobipartite.  Then in $O(m^2)$ time we can find a set of disjoint homogeneous pairs of cliques $\{(A_i,B_i)\}_{i=1}^k$ with the following property: If we construct a graph $G'$ from $G$ by contracting each $(A_i,B_i)$ down to an edge $a_ib_i$ and removing any subset of these new edges, then $G'$ is a circular interval graph or a composition of linear interval strips.
\end{corollary}

Cobipartite graphs are the degenerate subclass of quasi-line graphs in the same sense that graphs with $\alpha\leq 2$ are the degenerate subclass of claw-free graphs.  The fundamental reason for this is the fact that these degenerate graph classes are closed under the join operation.

To close this section, we give definitions of trigraphs, quasi-line and claw-free trigraphs, and homogeneous pairs of strong cliques.  In the next section, we define three useful types of homogeneous pairs of cliques, and analyze how they can intersect one another in a claw-free trigraph.  In Section \ref{sec:antithickening} we define thickenings, antithickenings, and optimal antithickenings, then prove that a non-degenerate claw-free trigraph has a unique optimal antithickening, and that we can find it efficiently.

\subsection{Trigraphs and homogeneous pairs of strong cliques}\label{sec:trigraphs}

Trigraphs are more general objects than graphs, and essentially differ from graphs by incorporating the possibility of two vertices being ``semiadjacent''.  They were first introduced in the first author's Ph.D.\ thesis \cite{chudnovskythesis}, and proved to be very useful in describing the structure of quasi-line and claw-free graphs \cite{clawfree7, clawfree5}.  Here we define trigraphs and generalize a variety of definitions from graphs to trigraphs; most of these generalizations are very natural.

A {\em trigraph} $G$ consists of a vertex set $V(G)$ and a function $\theta_G:V(G)^2 \rightarrow \{1,0,-1\}$ that defines the adjacency, with the properties
\begin{itemize*}
\item for any $v\in (G)$, $\theta_G(v,v)=0$,
\item for any $u,v\in V(G)$, $\theta_G(u,v)=\theta_G(v,u)$, and
\item for any distinct $u,v,w\in V(G)$, if $\theta_G(u,v)=0$ then $\theta_G(u,w)\neq 0$.
\end{itemize*}
Given vertices $u,v\in V(G)$, we say that $u$ and $v$ are {\em strongly adjacent} if $\theta_G(u,v)=1$, {\em adjacent} if $\theta_G(u,v)\in \{1,0\}$, {\em semiadjacent} if $\theta_G(u,v)=0$, {\em antiadjacent} if $\theta_G(u,v)\in \{0,-1\}$, and {\em strongly antiadjacent} if $\theta_G(u,v)=-1$.  The third property of $\theta_G$ tells us that semiadjacent pairs of distinct vertices (i.e.\ {\em semiedges}) form a matching.

The {\em complement of $G$}, denoted $\gbar$, is the trigraph on vertex set $V(G)$ in which for any $u,v\in V(G)$, $\theta_{\gbar}(u,v) = -\theta_{\gbar}(u,v)$.

Given disjoint vertex sets $A$ and $B$, we say that $A$ and $B$ are {\em complete} (resp.\ {\em strongly complete}) if for every $u\in A$ and $v\in B$, $u$ and $v$ are adjacent (resp.\ strongly adjacent).  We say that $A$ and $B$ are {\em anticomplete} (resp.\ {\em strongly anticomplete}) if for every $u\in A$ and $v\in B$, $u$ and $v$ are antiadjacent (resp.\ strongly antiadjacent).  If a vertex $v$ is strongly complete to $V(G)\setminus \{v\}$, we say $v$ is {\em universal}.  The {\em neighbourhood} (resp.\ {\em strong neighbourhood}) of a vertex $v$ is the set of vertices in $V(G)\setminus\{v\}$ that are adjacent (resp.\ strongly adjacent) to $v$.  The {\em antineighbourhood} (resp.\ {\em strong antineighbourhood}) of $v$ is the set of vertices in $V(G)\setminus\{v\}$ that are antiadjacent (resp.\ strongly antiadjacent) to $v$.  We use $N(v)$ to denote the neighbourhood of $v$.

A {\em clique} (resp.\ {\em strong clique}) is a set of pairwise adjacent (resp.\ strongly adjacent) vertices.  A {\em stable set} (resp.\ {\em strong stable set}) is a set of pairwise antiadjacent (resp.\ strongly antiadjacent) vertices.  The {\em stability number} of a trigraph $G$, denoted $\alpha(G)$, is the size of a maximum stable set.

Given distinct vertices $v_1,\ldots, v_4$ in a trigraph $G$, we say that $v_1v_2v_3v_4$ is a {\em square} (or a $C_4$) if the pairs $v_1v_3$ and $v_2v_4$ are antiadjacent, and the other four pairs are adjacent.  We say that there is a {\em claw} at a vertex $v$ if the neighbourhood of $v$ contains a stable set of size 3, and if $G$ contains no claw it is {\em claw-free}.  We say that $G$ is {\em cobipartite} if its vertices can be covered by two strong cliques, and if the neighbourhood of every vertex is cobipartite then $G$ is {\em quasi-line}.

For the remainder of the paper, we define a {\em non-degenerate} trigraph as a trigraph that is quasi-line and non-cobipartite, or claw-free with $\alpha \geq 3$; otherwise we say that the trigraph is {\em degenerate}.

A trigraph $G$ is {\em connected} if for any two distinct vertices $v$ and $v'$ there is a sequence $v=v_0,v_1,\ldots, v_k=v'$ such that for $0\leq i<k$, $v_i$ is adjacent to $v_{i+1}$.  In this paper we will restrict our attention to connected trigraphs.

A set $A$ of vertices in a trigraph $G$ is a {\em homogeneous set} if the following hold:
\begin{enumerate*}
\item $|V| > |A|\geq 2$.
\item Every vertex outside $A$ is either strongly complete or strongly anticomplete to $A$.
\end{enumerate*}
\noindent If the set $A$ is a strong clique, we say it is a {\em homogeneous strong clique}.

Suppose $A$ and $B$ are disjoint nonempty strong cliques in a trigraph $G$.  Then we say that $(A,B)$ is a {\em homogeneous pair of strong cliques} (HPOSC) if the following hold:
\begin{enumerate*}
\item $A$ and $B$ are not both singletons.
\item Every vertex outside $A\cup B$ is either strongly complete or strongly anticomplete to $A$, and is either strongly complete or strongly anticomplete to $B$.
\end{enumerate*}
As we will discuss in Section \ref{sec:antithickening}, the {\em thickening} operation expands semiedges in a trigraph into homogeneous pairs of strong cliques.

\section{Two types of homogeneous pairs of strong cliques}\label{sec:homo}

In this section we discuss two useful classes of HPOSC.  We define one now and the other later in the section.  Let $(A,B)$ be a HPOSC in a trigraph $G$.  We say that $(A,B)$ is {\em deletion-minimal} if $G|(A\cup B)$ contains a square, and no vertex in $A$ (resp.\ $B$) is strongly complete or strongly anticomplete to $B$ (resp.\ $A$).  This is the trigraph version of what is sometimes called a {\em proper} homogeneous pair of cliques \cite{faenzaos11}.  Our results on antithickenings of claw-free trigraphs rely on the fact that certain types of deletion-minimal HPOSCs intersect in a nice orderly way.  Before we get to proving this, we need a technical lemma.  Given a vertex set $X$ in a trigraph, we use $X^C$ to denote $V(G)\setminus X$.

\begin{lemma}\label{lem:homogeneousset}
Suppose a connected claw-free trigraph $G$ contains a homogeneous set $X$ that is not a strong clique.  Then $X$ is strongly complete to $X^C$, and $\alpha(G)=2$.  Furthermore if $G$ is quasi-line then $G|X$ is cobipartite.
\end{lemma}

This lemma implies that no connected non-degenerate trigraph contains a homogeneous set that is not a strong clique.

\begin{proof*}
We claim that $X$ is complete to $X^C$.  To see this, suppose the contrary.  Then since $G$ is connected, there are adjacent vertices $u$ and $v$ such that $u$ is strongly complete to $X$ and $v$ is strongly anticomplete to $X$.  Since $X$ is not a strong clique, $G$ contains a claw at $v$, a contradiction.  Thus $X$ is strongly complete to $X^C$.  Since $X^C$ is nonempty by the definition of a homogeneous set, $\alpha(G|X)=2$, and if $G$ is quasi-line then $G|X$ is cobipartite.

If there is a universal vertex, then clearly $\alpha(G)\leq 2$ and if $G$ is quasi-line, then $G|X$ (and indeed $G$) is cobipartite, so the lemma holds.  So assume there is no universal vertex in $G$.  It follows that $X^C$ must also be a homogeneous set which is not a strong clique, otherwise $G$ would contain a universal vertex in $X^C$.  Thus applying the symmetric argument tells us that $\alpha(G|X^C)=2$ and if $G$ is quasi-line, then $G|X^C$ is cobipartite.  Since $X$ is complete to $X^C$, the result follows.
\end{proof*}

Let $(A_1,B_1)$ and $(A_2,B_2)$ be HPOSCs in a trigraph.  We say that they have {\em skew intersection} if a part of one pair intersects both parts of the other pair, for example if $A_1\cap A_2$ and $A_1\cap B_2$ are both nonempty.

\begin{proposition}\label{prop:skew}
Suppose $G$ is a trigraph containing two deletion-minimal homogeneous pairs of strong cliques $(A_1,B_1)$ and $(A_2,B_2)$ with skew intersection.  Then $A_1\cup B_1\cup A_2\cup B_2$ can be covered by two strong cliques and is either a homogeneous set or the entire vertex set.
\end{proposition}

\begin{proof*}
Assume without loss of generality that $A_2$ intersects $A_1$ and $B_1$.  Since $(A_1,B_1)$ is deletion-minimal and $A_2$ is a strong clique, there is a vertex $b_1 \in B_1\setminus A_2$ that is antiadjacent to a vertex in $A_1\cap A_2$.  If $b_1 \notin B_2$ then the fact that $b_1$ is neither strongly complete nor strongly anticomplete to $A_2$ contradicts the fact that $(A_2,B_2)$ is a HPOSC, so $b_1\in B_2$.  Likewise there is a vertex $a_1 \in A_1\setminus A_2$ that is antiadjacent to a vertex in $B_1\cap A_2$, and $a_1 \in B_2$.  So $B_2$ must intersect both $A_1$ and $B_1$.

Suppose there is a vertex $v\in (A_2\cup B_2) \setminus (A_1\cup B_1)$.  Then $v$ is strongly complete to $A_1\cup B_1$.  It follows that $A_1\cup B_1\cup A_2\cup B_2$ can be covered by two strong cliques:  $A_1\cup (A_2 \setminus (A_1\cup B_1))$ and $B_1\cup (B_2 \setminus (A_1\cup B_1))$.

It remains for us to prove that $A_1\cup B_1\cup A_2\cup B_2$ is a homogeneous set or the entire vertex set, so assume there is a vertex $v \notin A_1\cup B_1\cup A_2\cup B_2$ and assume it is strongly complete to $A_1$.  It must therefore also be strongly complete to $A_2 \cup B_2$, and therefore also strongly complete to $B_1$.  If $v$ is not strongly complete to $A_1$ it is strongly anticomplete to $A_1$, and it follows that $v$ is strongly anticomplete to $A_2$, $B_2$, and $B_1$.  Thus $A_1\cup B_1\cup A_2\cup B_2$ is a homogeneous set or the entire vertex set.
\end{proof*}

\begin{figure}
\begin{center}
\includegraphics[scale=0.7]{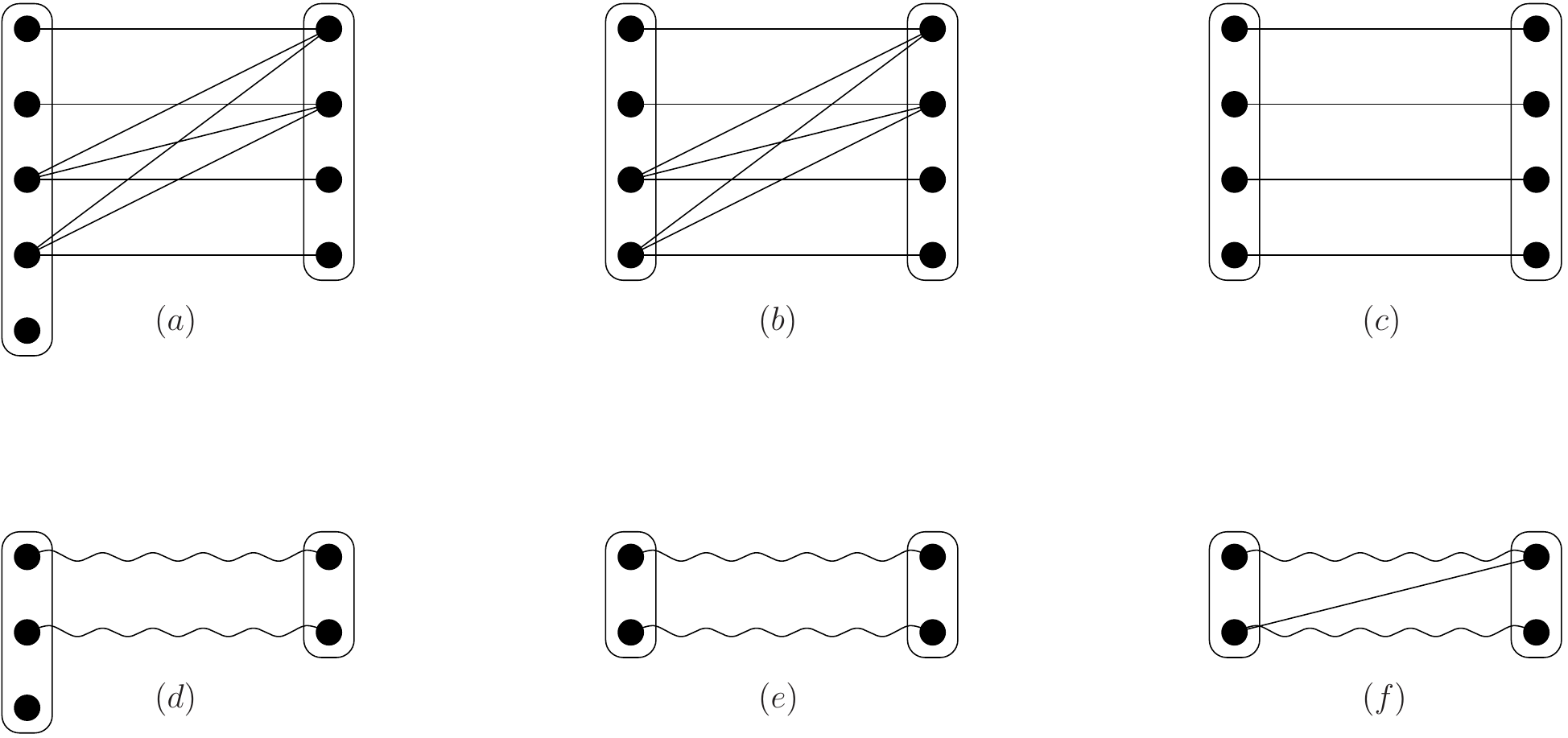}
\end{center}
\caption{\footnotesize{Some examples of homogeneous pairs of strong cliques.  Wiggly edges denote semiadjacency. Only $(c)$ and $(e)$ are square-connected, and only $(b)$, $(c)$, and $(e)$ are deletion-minimal.
}}
\label{fig:homogeneouspairs}
\end{figure}

\begin{corollary}\label{cor:skew}
Suppose $G$ is a connected claw-free trigraph containing two deletion-minimal homogeneous pairs of strong cliques $(A_1,B_1)$ and $(A_2,B_2)$ with skew intersection.  Then $G$ is degenerate. 
\end{corollary}
\begin{proof*}
This follows immediately from Lemma \ref{lem:homogeneousset} and Proposition \ref{prop:skew}.
\end{proof*}


\begin{proposition}
Suppose $G$ is a trigraph containing two deletion-minimal homogeneous pairs of strong cliques $(A_1,B_1)$ and $(A_2,B_2)$ without skew intersection, such that $A_1$ intersects $A_2$.  Then $B_1$ intersects $B_2$.
\end{proposition}

\begin{proof*}
Suppose to the contrary that $B_1\cap B_2=\emptyset$.  First assume that $A_1 = A_2$; this implies that no vertex in $B_1$ can be strongly complete or strongly anticomplete to $A_1=A_2$, a contradiction.  So assume without loss of generality that $A_2\setminus A_1 \neq \emptyset$.

Let $v$ be a vertex in $A_1\cap A_2$.  Since $v$ is in $A_1$ and $(A_1,B_1)$ is not deletion-minimal, there exist vertices $v_1$ and $v_2$ in $B_1$ such that $v$ is adjacent to $v_1$ and antiadjacent to $v_2$.  It follows that $v_1$ is strongly complete to $A_2$ and that $v_2$ is strongly anticomplete to $A_2$.  Since $A_2$ is not a subset of $A_1$, there is a vertex $u\in A_2\setminus A_1$, and it is strongly adjacent to $v_1$ and strongly antiadjacent to $v_2$.  This contradicts the fact that $(A_1,B_1)$ is a homogeneous pair of strong cliques.  
\end{proof*}

Now we know that deletion-minimal HPOSCs behave well in terms of intersection.  To get a better handle on inclusion-minimal deletion-minimal HPOSCs $(A,B)$ (meaning that if $A'\subseteq A$ and $B'\subseteq B$ and $(A',B')$ is a deletion-minimal HPOSC, then $(A',B')=(A,B)$), we introduce a new type of HPOSC.  We say that a homogeneous pair of strong cliques $(A,B)$ in a trigraph $G$ is {\em square-connected} if for any partition of $A$ (resp.\ $B$) into nonempty sets $A'$ and $A''$ (resp.\ $B'$ and $B''$), there is a square in $A\cup B$ intersecting both $A'$ and $A''$ (resp.\ $B'$ and $B''$).  As we will soon see, the smallest HPOSC containing a given square is always square-connected.  The proof of the following is an easy exercise:

\begin{proposition}
Every square-connected homogeneous pair of strong cliques $(A,B)$ is deletion-minimal.
\end{proposition}

Thus we can apply the results we have proven for deletion-minimal HPOSCs to our analysis of square-connected HPOSCs.  Further, the union of two intersecting square-connected HPOSCs is another square-connected HPOSC:

\begin{proposition}\label{prop:scunion}
Suppose $G$ is a trigraph containing two square-connected homogeneous pairs of strong cliques $(A_1,B_1)$ and $(A_2,B_2)$ without skew intersection and such that $A_1\cap A_2\neq \emptyset$ and $B_1\cap B_2\neq \emptyset$.  Then $(A_1\cup A_2, B_1\cup B_2)$ is a square-connected homogeneous pair of strong cliques.
\end{proposition}

\begin{proof*}
Clearly both $A=A_1\cup A_2$ and $B=B_1\cup B_2$ must be strong cliques.  Take some $v\notin A_1\cup B_1\cup A_2\cup B_2$; it is easy to see that $v$ is either strongly complete or strongly anticomplete to $A$, and either strongly complete or strongly anticomplete to $B$.  Thus $(A,B)$ is a homogeneous pair of strong cliques.

Suppose $(A,B)$ is not square-connected; assume without loss of generality that we can partition $A$ into nonempty $A'$ and $A''$ such that $A\cup B$ contains no square intersecting both $A'$ and $A''$.  Since $A_1$ and $A_2$ have nonempty intersection and $A'$ and $A''$ partition $A$, at least one of $A_1$ and $A_2$ intersects both $A'$ and $A''$; assume $A_1$ does.  This means that $A_1\cup B_1$ does not contain a square intersecting both $A'\cap A_1$ and $A''\cap A_1$, contradicting the fact that $(A_1,B_1)$ is square-connected.
\end{proof*}

We use the following straightforward algorithm of King and Reed \cite{kingr08} (generalized in the obvious way to trigraphs) to construct a square-connected homogeneous pair of cliques $(A,B)$ in a trigraph $G$, starting with strongly adjacent vertices $a_0$ and $a_1$ contained in a square in $G$.  Note that the other two vertices of this square must then be in $B$.

\begin{enumerate*}
\item Set $A=\{a_0,a_1\}$, set $B$ to be two vertices $b_0$, $b_1$ such that $a_0a_1b_1b_0$ is a square.  If no such vertices exist then quit.
\item If $A$ and $B$ are not strong cliques, there is no homogeneous pair of strong cliques with $\{a_0,a_1\}\subseteq A$, so quit.
\item If there is a vertex $v\notin A\cup B$ that is neither strongly complete nor strongly anticomplete to $A$, put $v$ in $B$ and go to Step 2.
\item If there is a vertex $v\notin A\cup B$ that is neither strongly complete nor strongly anticomplete to $B$, put $v$ in $A$ and go to Step 2.
\end{enumerate*}

We call this the {\em SCHPOSC algorithm}, for {\em square-connected homogeneous pair of strong cliques}.  It is easy to see that the output is unique, i.e.\ it does not depend on any ordering of the vertices.  We claim that it runs in $O(m)$ time for a given $a_0$ and $a_1$, and always returns the desired square-connected homogeneous pair of strong cliques if one exists.

First we note that if the algorithm adds a vertex $v$ to $(A,B)$ at some point, then $v$ must be in every homogeneous pair of strong cliques $(A,B)$ such that $a_0,a_1\in A$.  Thus it is clear that the algorithm finds the unique smallest homogeneous pair of strong cliques with $a_0,a_1$ in $A$.  It suffices to prove that the output is square-connected and computed efficiently.

\begin{proposition}
Given $a_0$ and $a_1$ contained in some square in a trigraph $G$, the SCHPOSC algorithm always produces a square-connected homogeneous pair of strong cliques $(A,B)$ with $A$ containing $a_0$ and $a_1$, if one exists.
\end{proposition}

\begin{proof*}
Let $(A_k,B_k)$ denote the partial pair constructed by the algorithm such that $|A_k\cup B_k|=k$.  It suffices to prove that $(A_k,B_k)$ is square-connected, which we will do by induction.  This is true by assumption for $k=4$, so let $k>4$ and assume without loss of generality that there is a vertex $v\in A_k\setminus A_{k-1}$.

Suppose $(A_k,B_k)$ is not square-connected.  It follows by the induction hypothesis (i.e.\ that $(A_{k-1},B_{k-1})$ is square-connected) that there is no square in $G|(A_k\cup B_k)$ that contains $v$.  Let $X$ and $Y$ 
partition $B_k$ such that $X$ is complete to $v$, $Y$ is anticomplete to $v$, and both $X$ and $Y$ are nonempty.  That is, $X$ will contain the strong neighbours of $v$ in $B_k$, $Y$ will contain the strong antineighbours of $v$ in $B_k$, and if there is a vertex $u\in B_k$ semiadjacent to $v$, we can put it in either $X$ or $Y$ to ensure that both are nonempty, which is possible because $|B_k|\geq 2$.

Since $(A_{k-1},B_{k-1})=(A_{k-1},B_{k})$ is square-connected, there is a square $S$ in $(A_{k-1},B_{k})$ intersecting both $X$ and $Y$.  This square contains a vertex $v'\in A_{k-1}$ that is adjacent to the vertex in $S\cap X$ and antiadjacent to the vertex in $S\cap Y$.  Replacing this vertex with $v$ gives us a square in $(A_k,B_k)$ containing $v$, thus it follows that $(A_k,B_k)$ is indeed square-connected.  The proposition follows.
\end{proof*}

\begin{proposition}
Given a strongly adjacent pair of vertices $a_0, a_1$ in a trigraph $G$, the SCHPOSC algorithm runs in $O(m)$ time.
\end{proposition}

\begin{proof*}
We only run Step 1 once, and for every pair of strongly adjacent vertices $b_0, b_1$ we can check if $\{a_0,a_1,b_1,b_0\}$ contains a square.  This takes constant time per strongly adjacent pair, so we can perform Step 1 in $O(m)$ time.  Having performed Step 1, and assuming we go on to Step 2, we partition the vertices outside $A\cup B$ into six sets.

Let $\tilde A$ (resp.\ $\tilde B$) be the vertices neither strongly complete nor strongly anticomplete to $B$ (resp.\ $A$).  The vertices in $\tilde A$ and $\tilde B$ must go into $A$ and $B$ respectively, and if $\tilde A$ and $\tilde B$ are not disjoint then we must quit, since every vertex $v$ must be strongly complete to either $A$ or $B$, or strongly anticomplete to both.  Computing these sets takes $O(m)$ time, since for each vertex $v$ we simply count the number of vertices in $A$ and $B$ to which $v$ is strongly adjacent.

Let $N_A$, $N_B$, $N_{AB}$ and $N_\emptyset$ be the sets of vertices that are, respectively: strongly complete to $A$ and strongly anticomplete to $B$, strongly complete to $B$ and strongly anticomplete to $A$, strongly complete to $A\cup B$, and strongly anticomplete to $A\cup B$.  Again, finding these sets takes $O(m)$ time, and as long as $\tilde A$ and $\tilde B$ are disjoint, we have partitioned the vertices into eight sets, including $A$ and $B$.

We must now advance through Steps 2, 3, and 4, updating our eight sets every time we move a vertex into $A$ or $B$.  Suppose we move a vertex $v$ from $\tilde A$ into $A$.  Since we know $\tilde A$ and $\tilde B$ are disjoint, we already know that $v$ is strongly complete to $A$.  When we do this, $\tilde A$ loses $v$ and no vertex moves into $\tilde A$.

Every vertex in $N_B$ that is adjacent to $v$ must move into $\tilde B$, and if a vertex $u$ in $N_\emptyset$ is adjacent to $v$ then we must quit, since $u$ is strongly anticomplete to $B$ and neither strongly complete nor strongly anticomplete to $A$.  We can clearly perform this part of the update in time proportional to the number of vertices adjacent to $v$.  The next part is slightly more subtle:  Every vertex in $N_{AB}$ that is antiadjacent to $v$ must move to $\tilde B$, and if a vertex in $N_A$ is antiadjacent to $v$, then we must quit.  We can perform this step in $O(|N_{AB}|+|N_A|)$ time, but we move every vertex that is antiadjacent to $v$ from $N_{AB}$ into $\tilde B$, or quit due to a vertex in $N_A$ being antiadjacent to $v$, at most once per vertex throughout the entire process.  Thus the time we waste during this step, i.e.\ the time spent inspecting vertices that we do not move, is at most proportional to the number of vertices adjacent to $v$.  Since the case of moving a vertex from $\tilde B$ into $B$ is symmetric, it follows that the entire update process, summed over every movement of a vertex into $A\cup B$, takes $O(m)$ time.
\end{proof*}

We conclude the section with an unsurprising but useful technical lemma.

\begin{lemma}\label{lem:inside}
If $(A,B)$ is a deletion-minimal homogeneous pair of strong cliques in a trigraph $G$, then there is a square-connected homogeneous pair of strong cliques $(A',B')$ in $G$ such that $A'\subseteq A$ and $B' \subseteq B$.
\end{lemma}

\begin{proof*}
We know that $(A,B)$ contains a square; say it is $a_0a_1b_1b_0$ such that $a_0,a_1\in A$ and $b_0,b_1\in B$.  We run the SCHPOSC algorithm on $a_0a_1$; call its putative output $(A',B')$.  The fact that $(A,B)$ is a homogeneous pair of strong cliques tells us that no vertex outside $A\cup B$ will be added to $A'$ or $B'$ during the algorithm's run.  It follows that the algorithm will indeed output a homogeneous pair of strong cliques $(A',B')$, and we know that it will be square-connected.
\end{proof*}

\section{Thickenings and optimal antithickenings}\label{sec:antithickening}

Having considered the intersection structure of square-connected and deletion-minimal HPOSCs, we now consider how we might best reduce or eliminate them.  In this section we answer this question by proving that every claw-free trigraph (and therefore every claw-free graph) has a unique optimal antithickening, which we can find in $O(m^2)$ time by simultaneously contracting a collection of disjoint inclusion-maximal square-connected HPOSCs.

We begin by defining {\em thickenings} of trigraphs.  We say that a trigraph $G$ is a {\em thickening} of another trigraph $G'$ if we can construct $G$ from $G'$ in the following way:
\begin{enumerate*}
\item $V(G)$ is a collection of nonempty disjoint strong cliques $\{I(v)\mid v\in V(G')\}$.
\item If two vertices $u,v\in V(G')$ are strongly adjacent, then $I(v)$ is strongly complete to $I(u)$ in $G$.
\item If two vertices $u,v\in V(G')$ are strongly antiadjacent, then $I(v)$ is strongly anticomplete to $I(u)$ in $G$.
\item If two vertices $u,v\in V(G')$ are semiadjacent, then $I(v)$ is neither strongly complete nor strongly anticomplete to $I(u)$ in $G$.
\end{enumerate*}
In this case we say that the function $I:V(G')\rightarrow 2^{V(G)}$ is a thickening from $G'$ to $G$, where $2^{V(G)}$ denotes the power set of $V(G)$.  For every two vertices $u,v \in V(G')$ such that no semiedge has exactly one of $u,v$ as an endpoint, either $|I(u)|=|I(v)|=1$ or $(I(u),I(v))$ is a homogeneous pair of strong cliques in $G$.  If $G$ is a thickening of $G'$ then we say that $G'$ is an {\em antithickening} of $G$.  Thickenings generalize structural operations used by Chv\'atal and Sbihi \cite{chvatals88} and Maffray and Reed \cite{maffrayr99} in the characterization of Berge quasi-line graphs\footnote{Every Berge claw-free graph is quasi-line.}; we refer the reader to Chapter 5 of \cite{kingthesis} for the whole story.

\begin{observation}
The relation of being a thickening is transitive: if $G$ is a thickening of $G'$ and $G'$ is a thickening of $G''$, then $G$ is a thickening of $G''$.
\end{observation}

Thickenings tell us how homogeneous pairs of strong cliques arise in $G$: some trivial pairs will arise from pairs of vertices in $V(G')$ not incident to semiedges (i.e.\ the HPOSC will consist of two strong cliques that are either strongly complete or strongly anticomplete to one another), and some more interesting pairs will arise from semiedges in $G'$.  Actually, given a trigraph $G$, we want to find an antithickening $G'$ of $G$ such that {\em all} interesting HPOSCs in $G$ arise from semiedges of $G'$ via thickening.

\subsection{Optimal antithickenings}

The aim of antithickenings in the context of the structure theorem for quasi-line graphs and trigraphs, particularly its algorithmic applications, is to get rid of all {\em nonlinear} homogeneous pairs of strong cliques (i.e.\ $(A,B)$ such that $G|(A\cup B)$ contains a square).  So we seek antithickenings that serve this purpose while preserving the structure of the original graph or trigraph as much as possible.

This guiding principle makes the desired criteria of an optimal antithickening clear.  First, we insist that the antithickening $G'$ does not contain a nonlinear homogeneous pair of strong cliques.  Our other criterion simply guarantees that the antithickening expresses the structure of $G$ as accurately as possible, meaning that we cannot further refine our reduction of square-connected homogeneous pairs of strong cliques without violating the first criterion.

We say that a trigraph $G$ is {\em laminar} if it contains no square-connected homogeneous pair of strong cliques.  Lemma \ref{lem:inside} tells us that $G$ is laminar precisely if it contains no deletion-minimal homogeneous pair of strong cliques.  Given a trigraph $G$, we say that an antithickening $G'$ of $G$ is an {\em optimal antithickening} if:
\begin{enumerate*}
\item $G'$ is laminar, and subject to that,
\item $|V(G')|$ is maximum.
\end{enumerate*}

Our aim in this section is to prove:

\begin{theorem}\label{thm:mainantithickening}
Every connected non-degenerate trigraph has a unique optimal antithickening, which we can find in $O(m^2)$ time.
\end{theorem}

We will prove that if $G'$ is an optimal antithickening of $G$, then for every semiedge $ab$ of $G'$, either $I(a)$ and $I(b)$ are singletons (and consequently their vertices must be semiadjacent), or $(I(a),I(b))$ is a square-connected homogeneous pair of cliques.  In fact, we will grow these pairs by using the SCHPOSC algorithm and combining pairs that intersect.  To see that we must assume the trigraph is non-degenerate, observe that a cycle of four strongly adjacent vertices (i.e.\ $C_4$ as a trigraph) has two optimal antithickenings.  Similarly, a copy of $C_4$ strongly complete to a copy of $C_5$ is claw-free, non-cobipartite, and has two optimal antithickenings.  It is no coincidence that these antithickenings arise from two HPOSCs with skew-intersection.

Our first lemma tells us that a square-connected homogeneous pair of cliques is always reduced to a semiedge in an optimal antithickening:

\begin{lemma}\label{lem:contained}
Let $G$ be a connected non-degenerate trigraph, and let $(A,B)$ be a square-connected homogeneous pair of cliques in $G$.  Suppose $G'$ is an optimal antithickening of $G$.  Then there exists a semiedge $ab$ of $G'$ such that $A\subseteq I(a)$ and $B\subseteq I(b)$.
\end{lemma}

\begin{proof*}
Let $A' = \{a_1,\ldots,a_k\}$ be the set of vertices of $G'$ for which $A$ intersects $I(a_i)$ and let $B' = \{b_1,\ldots,b_\ell\}$ be the set of vertices of $G'$ for which $B$ intersects $I(b_i)$.  We may assume for a contradiction that $k\geq 2$.

First we must prove that $A'\cap B' = \emptyset$, so suppose $a_1=b_1$, let $v_A^1$ be a vertex in $A\cap I(a_1)$, and let $v_B^1$ be a vertex in $B\cap I(b_1) = B\cap I(a_1)$.  Since $(A,B)$ is deletion-minimal, $v_A^1$ is not strongly complete to $B$ and $v_B^1$ is not strongly complete to $A$.  It follows that neither $A'$ nor $B'$ is a strong clique, and moreover there must be some $i>1$ for which $a_1a_i$ is a semiedge.  Likewise there must be some $j>1$ for which $b_1b_j$ is a semiedge.  Since the semiedges of $G'$ form a matching, $a_2=b_2$.  Furthermore there must be a vertex $v_A^2$ in $A\cap I(a_2)$ and a vertex $v_B^2$ in $B\cap I(b_2)=B\cap I(a_2)$ such that $v_A^1v_A^2v_B^2v_B^1$ is a square.  Moreover, $(I(a_1),I(a_2))$ is a deletion-minimal homogeneous pair of strong cliques in $G$, and it has skew-intersection with $(A,B)$.  It follows from Corollary \ref{cor:skew} that $G$ is degenerate, a contradiction.  Therefore $A'$ and $B'$ are disjoint.

We now claim that $A'$ and $B'$ are strong cliques.  Since for distinct $i,j$,  $I(a_i)$ and $I(a_j)$ cannot be strongly anticomplete, it is clear that $A'$ is a clique; $B'$ is a clique for the same reason.  So suppose $a_1a_2$ is a semiedge.  Since $(A,B)$ is deletion-minimal, $a_1$ is neither strongly complete nor strongly anticomplete to $B'$.  And since semiedges form a matching, there are no semiedges between $\{a_1,a_2\}$ and $B'$.  Thus no vertex in $I(a_1)\cup I(a_2)$ is strongly complete or strongly anticomplete to $B$.  Therefore $I(a_1)\subseteq A$ and $I(a_2)\subseteq A$, a contradiction since $A$ is a strong clique.  So $A'$ is a strong clique and likewise so is $B'$.  

We further claim that $(A',B')$ is actually a homogeneous pair of strong cliques.  First suppose a vertex $v$ in $V(G)\setminus (A'\cup B')$ is semiadjacent to a vertex of $A'\cup B'$, say $a_1$.  Then since $(A,B)$ is deletion-minimal, there exist $i,j$ such that $a_1$ is strongly adjacent to $b_i$ and strongly antiadjacent to $b_j$, which implies that $I(a_1)\subseteq A$, a contradiction since $I(v)$ is neither complete nor anticomplete to $I(a_1)$.  Therefore no semiedge in $G'$ has exactly one endpoint in $A'\cup B'$.  It follows easily that $(A',B')$ is a homogeneous pair of strong cliques.  We now need only prove that it contains a square, since this would contradict the assumption that $G'$ is an optimal antithickening of $G$ and therefore laminar.

The first task is to prove that $\ell\geq 2$, so assume that $\ell=1$.  Since $(A,B)$ is deletion-minimal, it follows that $b_1$ must be semiadjacent to a vertex in $A'$, so we may assume that $a_1b_1$ is a semiedge.  But since the semiedges of $G'$ form a matching, $I(a_2)$ is strongly complete or strongly anticomplete to $B$, a contradiction.  Therefore $\ell \geq 2$.

We claim that there exist $i$, $i'$, $j$, and $j'$ such that some square in $G$ intersects $I(a_i)$, $I(a_{i'})$, $I(b_j)$, and $I(b_{j'})$.  To see this, first note that since $(A,B)$ is square-connected, there is a square in $A\cup B$ intersecting both $A\cap I(a_1)$ and $A\setminus I(a_1)$; assume without loss of generality that it intersects $I(a_1)$ and $I(a_2)$.  Therefore there exist $j$ and $j'$ such that $a_1$ is adjacent to $b_j$ and antiadjacent to $b_{j'}$, and $a_2$ is adjacent to $b_{j'}$ and antiadjacent to $b_j$.  To see that $j\neq j'$, note that if $j=j'$ then $b_j$ is semiadjacent to both $a_i$ and $a_{i'}$, which is impossible.  The claim follows, so $A'\cup B'$ contains a square, and therefore contains a square-connected pair of cliques, contradicting the fact that $G'$ must be laminar.
\end{proof*}

\begin{corollary}
Let $G$ be a connected non-degenerate trigraph containing two square-connected homogeneous pairs of strong cliques $(A_1,B_1)$ and $(A_2,B_2)$ without skew intersection.  If $A_1\cap A_2\neq \emptyset$ and $B_1\cap B_2\neq \emptyset$ and $G'$ is an optimal antithickening of $G$, then there exists a semiedge $ab$ of $G'$ such that $A_1\cup A_2 \subseteq I(a)$ and $B_1\cup B_2 \subseteq I(b)$.
\end{corollary}
\begin{proof*}
This follows immediately from Lemma \ref{lem:contained} and Proposition \ref{prop:scunion}.
\end{proof*}

We now want to show that every set that reduces to a semiedge $ab$ in an optimal antithickening is either a semiedge or a SCHPOSC.  Our approach is to take $(I(a),I(b))$, prove that it is deletion-minimal and therefore contains a square, build a minimal SCHPOSC $(A,B)$ containing this square, and finally prove that all of $(I(a),I(b))$ must be a SCHPOSC.  Proving that $(I(a),I(b))$ is deletion-minimal is tedious enough that we wish to separate it from the rest of the proof.

\begin{lemma}\label{lem:imageisdm}
Let $G$ be a connected non-degenerate trigraph and let $G'$ be an optimal antithickening of $G$.  Then for any semiedge $ab$ of $G'$, $(I(a),I(b))$ is a deletion-minimal homogeneous pair of strong cliques or a semiedge.
\end{lemma}
\begin{proof}
Clearly if $I(a)$ and $I(b)$ are both singletons then the vertices are semiadjacent, so assume that $|I(a)|+|I(b)|\geq 3$.  Thus $(I(a),I(b))$ is a homogeneous pair of strong cliques; to reach a contradiction we assume it is not deletion-minimal.

First suppose that $I(a)$ contains a vertex $u$ that is strongly complete to $I(b)$.  We will show that we can remove $u$ from $I(a)$ and contradict the optimality of $G'$.

Let $G''$ be the trigraph constructed from $G'$ by adding a vertex $a'$ and making it strongly adjacent to $a$, $b$, and every vertex that is strongly adjacent to $a$ in $G'$, and strongly antiadjacent to everything else (see Figure \ref{fig:optimalantithickening}(a)) (note that $a$ is not semiadjacent to any vertex other than $b$).  It is easy to see that $G''$ is an antithickening of $G$ with $I(a')=\{u\}$.  It has more vertices than $G'$, so to contradict the optimality of $G'$ it suffices to show that $G''$ is laminar.

\begin{figure}
\begin{center}
\includegraphics[scale=0.7]{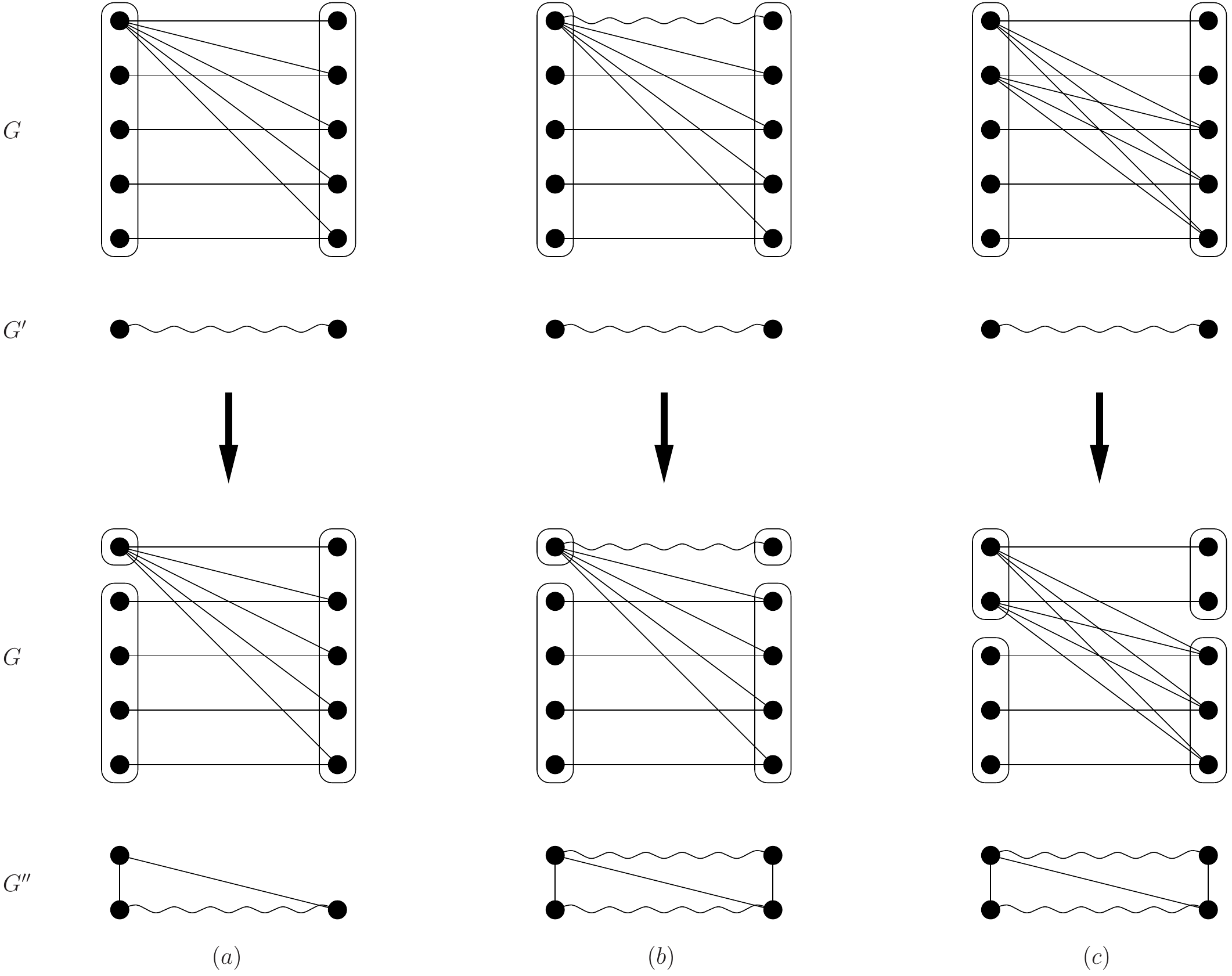}
\end{center}
\caption{\footnotesize{Three ways of refining an antithickening.}}
\label{fig:optimalantithickening}
\end{figure}

To show this, it suffices to prove that $G''$ does not contain a deletion-minimal homogeneous pair of strong cliques $(A'',B'')$, so suppose one exists.  Clearly $a'$ is contained in $A''$ or $B''$, otherwise $(A'',B'')$ would be a deletion-minimal homogeneous pair of strong cliques in $G'$.  So assume $a'\in A''$.  Since $(A'',B'')$ is deletion-minimal and $a'$ is not the endpoint of a semiedge, $a'$ is strongly adjacent to some vertex in $B''$ and is strongly antiadjacent to another vertex in $B''$.  It follows that $a$ is also in $A''$, and therefore $b$ must be in $B''$.  It follows easily from the definition of $G''$ that $(A''\setminus \{a'\}, B'')$ is a deletion-minimal homogeneous pair of strong cliques in $G'=G''-a'$, contradicting the assumption that it is an optimal antithickening of $G$. The case in which $u$ is anticomplete to $I(b)$ follows in the same way; we omit the details.

We must now consider the case in which $(I(a),I(b))$ is not deletion-minimal, no vertex in $I(a)$ is strongly complete or strongly anticomplete to $I(b)$, and no vertex in $I(b)$ is strongly complete or strongly anticomplete to $I(a)$.  Clearly neither $I(a)$ nor $I(b)$ is a singleton.  By the definition of a deletion-minimal homogeneous pair of strong cliques, $G|(I(a)\cup I(b))$ does not contain a square.  We claim that there must exist (without loss of generality) some $a_1\in I(a)$ and $b_1 \in I(b)$ such that $a_1b_1$ is a semiedge, $a_1$ is strongly complete to $I(b)\setminus \{b_1\}$, and $b_1$ is strongly anticomplete to $I(a)\setminus \{a_1\}$.

To see this, assume without loss of generality that there exist strongly adjacent $a_1\in A$ and $b_2\in B$; choose $a_1\in I(a)$ maximizing the number of vertices in $I(b)$ strongly adjacent to $a_1$, and choose $b_1\in I(b)$ antiadjacent to $a_1$.  Suppose $b_1$ is strongly antiadjacent to $I(a)\setminus\{a_1\}$.  Then $a_1$ and $b_1$ are semiadjacent, and if the claim is not true then there exists $b_3\in I(b)$ strongly antiadjacent to $a_1$; there must exist $a_2 \in I(a)$ adjacent to $b_3$, giving us a square and a contradiction.  So there must exist $a_2\in I(a)\setminus \{a_1\}$ adjacent to $b_2$.  By our choice of $a_1$, there must be some $b_3$ strongly adjacent to $a_1$ but antiadjacent to $b_2$, since $b_1$ cannot be semiadjacent to both $a_1$ and $a_2$.  Again this gives us a square and a contradiction, proving the claim.

Construct $G''$ by adding vertices $a'$ and $b'$ to $G'$ such that $a'b'$ is a semiedge, $a'$ is strongly adjacent to $a$ and precisely those vertices in $V(G')\setminus \{a\}$ that are adjacent to $a$, and $b'$ is strongly adjacent to $b$ and strongly antiadjacent to precisely those vertices in $V(G')\setminus \{b\}$ that are antiadjacent to $b$ (see Figure \ref{fig:optimalantithickening}$(b)$).  Clearly $G''$ is an antithickening of $G$, so suppose it contains a square-connected homogeneous pair of strong cliques $(A'',B'')$.  Since $G'$ is laminar, this pair must contain one and therefore both of $a'$ and $b'$.  And since $G''-a-b$ is isomorphic to $G'$, the pair must contain one and therefore both of $a$ and $b$.  We can therefore assume $\{a,a'\}\subseteq A''$ and $\{b,b'\}\subseteq B''$.

It remains to contradict the fact that $G'$ is laminar by proving that $(A''\setminus\{a'\},B''\setminus\{b'\})$ is a deletion-minimal homogeneous pair of strong cliques; first we must prove that it contains a square.  Since $(A'',B'')$ is square-connected, there is a square in $(A'',B'')$ containing $a$.  Observe that this square contains at most one of $b$ and $b'$, so if it does not contain $b'$ then we are done; assume it contains $b'$ but not $b$.  Then removing $b'$ from this square and replacing it with $b$ gives us another square, so indeed $(A''\setminus\{a'\},B''\setminus\{b'\})$ contains a square.  Now suppose without loss of generality that some vertex $v$ in $A''\setminus\{a'\}$ is strongly complete or strongly anticomplete to $B''\setminus\{b'\}$.  Clearly $v\neq a$, but this means that $v$ is strongly complete or strongly anticomplete to $B''$, contradicting the fact that $(A'',B'')$ is deletion-minimal.  Thus it indeed follows that $G'$ is not laminar, a contradiction.
\end{proof}

We need a few more routine results before describing our antithickening algorithm.

\begin{lemma}\label{lem:vertexissc}
Let $(A,B)$ be a homogeneous pair of strong cliques in a trigraph $G$, and let $(A_1,B_1)$ be a square-connected homogeneous pair of strong cliques in $G$ such that $A_1\subset A$ and $B_1\subseteq B$.  If there is no square intersecting both $A_1\cup B_1$ and $(A\cup B)\setminus (A_1\cup B_1)$, then every vertex of $A\setminus A_1$ is strongly complete or strongly anticomplete to $B_1$, and every vertex of $B\setminus B_1$ is either strongly complete or strongly anticomplete to $A_1$.
\end{lemma}

\begin{proof}
Assume that no square in $(A,B)$ intersects both $A_1\cup B_1$ and $(A\cup B)\setminus (A_1\cup B_1)$.  Now suppose some $v$ in $A\setminus A_1$ is adjacent to a vertex of $B_1$ and strongly antiadjacent to a vertex of $B_1$.  Then let $B_1'$ denote the set of vertices in $B_1$ adjacent to $v$, and let $B_1''$ denote $B_1\setminus B_1''$, and note that both sets are nonempty.  Since $(A_1,B_1)$ is square-connected, there is a square in $A_1\cup B_1$ intersecting both $B_1'$ and $B_1''$, and it follows easily that there is a square in $(A,B)$ intersecting both $A_1$ and $A\setminus A_1$.  The same argument applies if $v$ is antiadjacent to a vertex of $B_1$ and strongly adjacent to a vertex of $B_1$, so indeed $v$ is either strongly complete or strongly anticomplete to $B_1$.

The symmetric argument tells us that every vertex of $B\setminus B_1$ is either strongly complete or strongly anticomplete to $A_1$.
\end{proof}

\begin{lemma}\label{lem:squareissc}
Let $(A,B)$ be a homogeneous pair of strong cliques in a trigraph $G$, and let $(A_1,B_1)$ and $(A_2,B_2)$ be disjoint square-connected homogeneous pairs of strong cliques in $G$ such that $A_1\cup A_2\subseteq A$ and $B_1\cup B_2\subseteq B$.  Suppose there is no square intersecting both $A_1\cup B_1$ and $A_2\cup B_2$.  Then either
\begin{itemize}
\item $A_2$ is strongly complete to $B_1$ and $B_2$ is strongly anticomplete to $A_1$, or
\item $A_2$ is strongly anticomplete to $B_1$ and $B_2$ is strongly anticomplete to $A_1$.
\end{itemize}
\end{lemma}
\begin{proof}
It follows easily from Lemma \ref{lem:vertexissc} that every vertex in $A_2$ is either strongly complete or strongly anticomplete to $B_1$.  Suppose there exist $a,a'\in A_2$ such that $a$ is strongly complete to $B_1$ and $a'$ is strongly anticomplete to $B_1$.  Now let $b,b'\in B_2$ be vertices such that $aa'b'b$ is a square.  Then observe that for any $b''\in B_1$, $aa'b'b''$ is a square, a contradiction.  Thus $A_2$ is either strongly complete or strongly anticomplete to $B_1$.  By the symmetric argument, $B_2$ is either strongly complete or strongly anticomplete to $A_1$.  Assume without loss of generality that $A_2$ is strongly complete to $B_1$; it remains to show that $B_2$ is strongly anticomplete to $A_1$, so suppose to the contrary that $B_2$ is strongly complete to $A_1$.  Thus there cannot exist an antiadjacent pair of vertices in each of  $(A_1,B_1)$ and $(A_2,B_2)$, which of course contradicts the fact that both pairs are square-connected.
\end{proof}

\begin{corollary}\label{cor:imageissc}
Let $G$ be a connected non-degenerate trigraph and let $G'$ be an optimal antithickening of $G$.  Then for any semiedge $ab$ of $G'$, $(I(a),I(b))$ is a square-connected homogeneous pair of strong cliques or a semiedge.
\end{corollary}

\begin{proof}
Suppose that $(I(a),I(b))$ is not square-connected.  Lemma \ref{lem:imageisdm} tells us that $(I(a),I(b))$ is deletion-minimal.
Let $(A_1,B_1), (A_2,B_2),\ldots, (A_k,B_k)$ be a set of pairwise disjoint square-connected homogeneous pairs of strong cliques such that $A_1\cup\ldots\cup A_k \subseteq A$ and $B_1\cup\ldots\cup B_k\subseteq B$ such that every square in $A\cup B$ lies completely inside some $A_i\cup B_i$.  The fact that $(A_1,B_1),\ldots,(A_k,B_k)$ exist follows from the SCHPOSC algorithm and Proposition \ref{prop:scunion}.  Since $(A,B)$ is not square-connected, $A_1\cup B_1\neq A\cup B$.  Let the vertices of $A\setminus (\cup_{i}A_i)$ be $a'_1,\ldots, a'_{\ell_1}$ and let the vertices of $B\setminus (\cup_{i}B_i)$ be $b'_1,\ldots, b'_{\ell_2}$.

Informally speaking, we will construct an antithickening $G''$ of $G$ by using the same antithickening as between $G$ and $G'$ except on $A$ and $B$.  On $A$ and $B$ we will proceed by simply contracting each $(A_i,B_i)$ down to a semiedge.  We then prove that $G''$ is laminar.  But first we must give a more formal description of $G''$.

Let $N(a)$ and $\overline N(a)$ denote the strong neighbourhood and strong antineighbourhood of $a$ in $G'$, respectively.  Let $N(b)$ and $\overline N(b)$ denote the strong neighbourhood and strong antineighbourhood of $b$ in $G'$, respectively.  We construct the trigraph $G''$ from $G'-a-b$ as follows.

\begin{itemize}
\item For $i=1,\ldots,k$, we add semiadjacent $a_i$ and $b_i$, with $a_i$ strongly adjacent to $N(a)$ and strongly antiadjacent to $\overline N(a)$, and with $b_i$ strongly adjacent to $N(b)$ and strongly antiadjacent to $\overline N(b)$.
\item For $i=1,\ldots,\ell_1$, we add a vertex $a'_i$ strongly adjacent to $N(a)$ and strongly antiadjacent to $\overline N(a)$.
\item For $i=1,\ldots,\ell_2$, we add a vertex $b'_i$ strongly adjacent to $N(b)$ and strongly antiadjacent to $\overline N(b)$.
\item $A'=\{a_1,\ldots,a_k,a'_1,\ldots,a'_{\ell_1}\}$ is a strong clique, and so is $B'= \{b_1,\ldots,b_k,b'_1,\ldots,b'_{\ell_2}\}$.
\item $a_i$ and $b_j$ are strongly adjacent if $A_i$ is strongly complete to $B_j$, and strongly anticomplete otherwise.
\item $a_i$ and $b'_j$ are strongly adjacent if $b'_j$ is strongly complete to $A_i$, and strongly anticomplete otherwise.
\item $b_j$ and $a'_i$ are strongly adjacent if $a'_i$ is strongly complete to $B_j$, and strongly anticomplete otherwise.
\item $a'_i$ and $b'_j$ are strongly adjacent if $a'_i$ and $b'_j$ are strongly adjacent in $G$, semiadjacent if they are semiadjacent in $G$, and strongly antiadjacent if they are strongly antiadjacent in $G$. 
\end{itemize}
It is straightforward to confirm that $G'$ is an antithickening of $G''$, and that $G''$ is an antithickening of $G$.  Clearly $G''$ has more vertices than $G'$, so it suffices for us to prove that $G''$ is laminar.

Suppose there is a square-connected homogeneous pair of strong cliques $(A'',B'')$ in $G''$.  It is straightforward to confirm that $(A'',B'')$ must lie inside $A'\cup B'$.  Observe that both $G'$ and $G''$ must be non-degenerate, so we can assume $A''\subseteq A'$ and $B''\subseteq B'$.  It suffices, then, to prove that there is no square in $A'\cup B'$.  So suppose one exists, say $S=uvxy$ in cycle order with $u,v\in A'$ and $x,y\in B'$.

First suppose $S$ contains $a_1$ and $b_1$; assume $u=a_1$ and $y=b_1$.  Then it is easy to see that whether or not $v$ and $x$ appear in sets $A_i$ and/or $B_j$, there is a square $S'$ in $G$ that is partially but not completely contained in $A_1\cup B_1$, contradicting our choice of $(A_1,B_1)$.  The same conclusion applies if $u=a_1$ and $x=b_1$.  Therefore since at most one vertex of $S$ appears in any pair $(A_i,B_i)$, and there is no square in $(A\setminus (\cup_{i}A_i))\cup(B\setminus (\cup_{i}B_i))$, it follows that there is a square in $(A,B)$ partially but not completely contained in $A_i\cup B_i$ for some $i$, a contradiction.
\end{proof}

\subsection{Uniqueness}

Before we give an efficient algorithm for finding the optimal antithickening of a quasi-line trigraph, we must first prove uniqueness.

\begin{theorem}
Any connected non-degenerate trigraph $G$ has a unique optimal antithickening.
\end{theorem}

\begin{proof}
We proceed by induction on $k=|V(G)|$.  The theorem is clearly true for trigraphs on at most three vertices.  So assume the theorem holds for trigraph on fewer than $k$ vertices.

If $G$ is laminar then clearly $G$ is its own (unique) optimal antithickening and we are done.  Therefore by Lemma \ref{lem:inside} there is a square-connected homogeneous pair of cliques $(A,B)$ in $G$.  Let $\tilde G$ be the antithickening of $G$ reached by contracting $A$ into a vertex $\tilde a$ and contracting $B$ into a vertex $\tilde b$; note that $\tilde a \tilde b$ is a semiedge in $\tilde G$.

Inductively, $\tilde G$ has a unique optimal antithickening; call it $\tilde G'$ and note that $\tilde G'$ is also an antithickening of $G$.  Let $G'$ be any optimal antithickening of $G$; it suffices to show that $G'$ is an antithickening of $\tilde G$.
 Lemma \ref{lem:contained} tells us that there are semiadjacent vertices $a,b$ in $G'$ such that in $G$, $A\subseteq I(a)$ and $B\subseteq I(b)$.  Therefore we can define a thickening $I':V(G')\rightarrow 2^{V(\tilde G)}$ from $G'$ to $\tilde G$ as follows.  For $v\in V(G')$ not in $\{a,b\}$, set $I'(v)$ be equal to $I(v)$, i.e.\ the same as in the thickening from $G'$ to $G$.  Now set
\begin{eqnarray*}
I'(a) &=& (I(a)\setminus A ) \cup \{ \tilde a\}\\
I'(b) &=& (I(b)\setminus B ) \cup \{ \tilde b\}.
\end{eqnarray*}
It suffices to check that $I'$ is a valid thickening.  This is straightforward and we leave the details to the reader.
\end{proof}

\subsection{An algorithm}

We are now prepared to prove Theorem \ref{thm:mainantithickening}, which we restate.

\begin{theorem*}
Every connected non-degenerate trigraph has a unique optimal antithickening, which we can find in $O(m^2)$ time.
\end{theorem*}

\begin{proof}
Let $G$ be a connected non-degenerate trigraph.  Our algorithm is as follows.
\begin{enumerate*}
\item For each pair $uv$ of strongly adjacent vertices contained in a square in $G$, do the following:
\begin{enumerate*}
\item Run the SCHPOSC algorithm to find a square-connected homogeneous pair of strong cliques $(A_{uv},B_{uv})$ such that $\{u,v\}\subseteq A_{uv}$, if one exists.
\item If $A_{uv}$ intersects some previously constructed $A_{u'v'}$, then set $A_{u'v'}:= A_{u'v'}\cup A_{uv}$ and set $B_{u'v'}:= B_{u'v'}\cup B_{uv}$, and forget the sets $A_{uv}$ and $B_{uv}$.  Take the analogous action if $A_{uv}$ intersects $B_{u'v'}$ or if $B_{uv}$ intersects $A_{u'v'}$ or $B_{u'v'}$.
\end{enumerate*}
\item We now have a set of pairwise disjoint square-connected homogeneous pairs of cliques $\{(A_{u_iv_i},\allowbreak B_{u_iv_i})\}_{i=1}^k$.  To construct the antithickening $G'$, we contract each $A_{u_iv_i}$ and $B_{u_iv_i}$ into semiadjacent vertices $a_i$ and $b_i$.
\end{enumerate*}
We must now prove correctness and running time.  Bounding the running time is simple.  We already know that running the SCHPOSC algorithm for $uv$ takes $O(m)$ time.  We can easily detect intersection between $A_{uv}$ and $A_{u'v'}$ because at any given point, each vertex of $G$ is in at most one set $A_{u'v'}$.  So we can perform Step 1 in $O(m^2)$ time.  Since $k<|V(G)|$, we can easily run Step 2 in $O(nm)$ time.  Thus the overall running time is $O(m^2)$.

We now prove that the output trigraph $G'$ is the optimal antithickening $G''$.  Let $I$ be the thickening from $G'$ to $G$, and let $I'$ be the thickening from $G''$ to $G$.  From Proposition \ref{prop:scunion} we know that in Step 2(b), if $A_{uv}$ intersects $A_{u'v'}$ then $(A_{u'v'}\cup A_{uv}, B_{u'v'}\cup B_{uv})$ is a square-connected homogeneous pair of strong cliques.  Therefore each contracted pair $(A_{u_iv_i},B_{u_iv_i})$ is a square-connected homogeneous pair of strong cliques.  In particular, Lemma \ref{lem:contained} tells us that if there exist vertices $u$ and $v$ in $V(G)$ and $a$ in $V(G')$ such that $\{u,v\}\subseteq I(a)$, then there exists $a'$ in $V(G'')$ such that $\{u,v\}\in I'(a')$.

It remains to show the converse, so suppose there exist vertices $u$ and $v$ in $V(G)$ and $a'$ in $V(G'')$, such that $\{u,v\} \in I'(a')$.  Since $G'$ is optimal, $a'$ must be semiadjacent to some vertex $b'$, otherwise $u$ and $v$ would not be together in $I'(a')$.  By Corollary \ref{cor:imageissc}, $(I(a'),I(b'))$ is a square-connected homogeneous pair of strong cliques in $G$, so it follows from Lemma \ref{lem:contained} that there is some $a\in V(G')$ such that $\{u,v\}\subseteq I(a)$.  Therefore two vertices of $G$ have the same preimage in $G'$ if and only if they have the same preimage in $G''$; thus $G'$ must be isomorphic to $G''$, and $G'$ must be an optimal antithickening of $G$.
\end{proof}

\section{Conclusion}

While we do not rule out the possibility that the optimal antithickening can be found in $o(m^2)$ time, a faster algorithm would require a different approach:  Consider the graph consisting of two disjoint cliques of size $k$, connected by a perfect matching.  While the optimal antithickening consists of a single semiedge, the graph has $\Theta(k^2)=\Theta(m)$ minimal square-connected homogeneous pairs of cliques, each containing four vertices.  Therefore a direct approach using the SCHPOSC algorithm is condemned to take at least $\Theta(m^2)$ time.  

In a forthcoming paper, the second author will give an $O(nm)$-time algorithm for constructing an optimal strip decomposition of a composition of linear interval strips.  This algorithm is most useful when applied to laminar trigraphs, so the speed of our antithickening algorithm represents a bottleneck in the decomposition.  For these reasons, a faster antithickening algorithm would be both surprising and useful.

\section{Acknowledgements}

The authors are grateful to Gianpaolo Oriolo for useful discussions about this work.

\bibliography{masterbib}
\end{document}